\documentclass[11pt,a4paper,reqno]{amsart}%
\usepackage[latin1]{inputenc}
\usepackage{mathrsfs}
\usepackage{dsfont}
\usepackage{hyperref}
\usepackage{amsmath}
\usepackage{amssymb}
\usepackage{amsthm}
\usepackage{amsfonts}
\usepackage{amstext}
\usepackage{amsopn}
\usepackage{amsxtra}
\usepackage{mathrsfs}
\usepackage{dsfont}
\usepackage{graphicx}
\usepackage{color}

\makeatletter
\def\@tocline#1#2#3#4#5#6#7{\relax
  \ifnum #1>\c@tocdepth 
  \else
    \par \addpenalty\@secpenalty\addvspace{#2}%
    \begingroup \hyphenpenalty\@M
    \@ifempty{#4}{%
      \@tempdima\csname r@tocindent\number#1\endcsname\relax
    }{%
      \@tempdima#4\relax
    }%
    \parindent\z@ \leftskip#3\relax \advance\leftskip\@tempdima\relax
    \rightskip\@pnumwidth plus4em \parfillskip-\@pnumwidth
    #5\leavevmode\hskip-\@tempdima
      \ifcase #1
       \or\or \hskip 1em \or \hskip 2em \else \hskip 3em \fi%
      #6\nobreak\relax
      \dotfill
      \hbox to\@pnumwidth{\@tocpagenum{#7}}
    \par
    \nobreak
    \endgroup
  \fi}
\makeatother

\newtheorem{theorem}{Theorem}[section]
\newtheorem{proposition}[theorem]{Proposition}
\newtheorem{remark}[theorem]{Remark}
\newtheorem{lemma}[theorem]{Lemma}
\newtheorem{corollary}[theorem]{Corollary}

\numberwithin{equation}{section}

\newcommand{\ii}{\infty}
\newcommand\R{{\ensuremath {\mathbb R} }}
\newcommand\C{{\ensuremath {\mathbb C} }}
\newcommand\N{{\ensuremath {\mathbb N} }}

\newcommand\1{{\ensuremath {\mathds 1} }}

\newcommand\nn{\nonumber}

\renewcommand\phi{\varphi}
\newcommand{\bH}{\mathbb{H}}
\newcommand{\gH}{\mathfrak{H}}

\newcommand{\gS}{\mathfrak{S}}

\newcommand{\wto}{\rightharpoonup}

\newcommand{\cM}{\mathcal{M}}

\newcommand{\cF}{\mathcal{F}}
\newcommand{\gF}{\mathfrak{F}}
\newcommand{\cN}{\mathcal{N}}

\newcommand{\cH}{\mathcal{H}}

\newcommand{\eps}{\epsilon}

\renewcommand{\epsilon}{\varepsilon}
\newcommand\pscal[1]{{\ensuremath{\left\langle #1 \right\rangle}}}
\newcommand{\norm}[1]{ \left| \! \left| #1 \right| \! \right| }
\DeclareMathOperator{\tr}{{\rm Tr}}
\DeclareMathOperator{\Tr}{{\rm Tr}}
\newcommand{\supp}{{\rm Supp}}

\renewcommand{\ge}{\geqslant}
\renewcommand{\le}{\leqslant}
\renewcommand{\geq}{\geqslant}
\renewcommand{\leq}{\leqslant}

\title[Anharmonic oscillator based measures]{Gibbs measures based on 1D (an)harmonic oscillators as mean-field limits}

\author[M. Lewin]{Mathieu LEWIN}
\address{CNRS \& CEREMADE, Universit\'e Paris-Dauphine, PSL Research University, Place de Lattre de Tassigny, F-75016 PARIS, France} 
\email{mathieu.lewin@math.cnrs.fr}

\author[P.~T. Nam]{Phan Th\`anh NAM}
\address{Masaryk University, Department of Mathematics and Statistics, Kotl\'a\v rsk\'a 2, 
61137 Brno, Czech Republic} 
\email{ptnam@math.muni.cz}

\author[N. Rougerie]{Nicolas ROUGERIE}
\address{Universit\'e Grenoble 1 \& CNRS,  LPMMC (UMR 5493), B.P. 166, F-38042 Grenoble, France}
\email{nicolas.rougerie@grenoble.cnrs.fr}

\date{February, 2018}

\begin{document}

\begin{abstract}
We prove that Gibbs measures based on 1D defocusing nonlinear Schr\"odinger functionals with sub-harmonic trapping can be obtained as the mean-field/large temperature limit of the corresponding grand-canonical ensemble for many bosons. The limit measure is supported on Sobolev spaces of negative regularity and the corresponding density matrices are not trace-class. The general proof strategy is that of a previous paper of ours, but we have to complement it with Hilbert-Schmidt estimates on reduced density matrices.
\end{abstract}

\maketitle

\setcounter{tocdepth}{2}
\tableofcontents

\section{Introduction}

Gibbs measures based on nonlinear Schr\"odinger energy functionals play a central role in constructive quantum field theory (CQFT)~\cite{GliJaf-87,Simon-74,DerGer-13,VelWig-73} and in the low-regularity probabilistic Cauchy theory of nonlinear Schr\"odinger (NLS) equations~\cite{CacSuz-14,Bourgain-94,Bourgain-96,Bourgain-97,BouBul-14a,BouBul-14b,BurThoTzv-10,LebRosSpe-88,ThoOh-15,ThoTzv-10,Tzvetkov-08}. They also are the natural long-time asymptotes for nonlinear dissipative stochastic PDEs~\cite{BouDebFuk-17,PraDeb-03,RocZhuZhu-16,TsaWeb-16}. Recently, we have shown that, at least in the most well-behaved cases, they can be derived from the linear many-body quantum mechanical problem. Namely, many-body bosonic thermal equilibrium states converge in a certain mean-field/large-temperature limit~\cite{LewNamRou-14d,LewNamRou-ICMP,Rougerie-xedp15} to nonlinear Gibbs measures (see the recent~\cite{FroKnoSchSoh-17} for a corresponding time-dependent statement). The goal of this note is to extend this result to the case of somewhat less well-behaved measures, e.g. those based on the 1D harmonic oscillator studied in~\cite{BurThoTzv-09,BurThoTzv-10,BouDebFuk-17}.

\medskip

Consider the NLS flow on $\R ^{d+1}$
\begin{equation}\label{eq:NLS}
 i\partial_t u = -\Delta u + V u + \left(w * |u| ^2\right) u, 
\end{equation}
with $V$ a trapping potential and $w$ an interaction potential (say a delta function). A natural candidate for an invariant measure under~\eqref{eq:NLS} can be defined formally in the manner 
$$ \mu (du) = \frac{1}{z_r} \exp\left(- \frac{1}{2} \iint_{\R^d \times \R^d} |u(x)| ^2 w (x-y) |u(y)| ^2 dxdy \right) \mu_0 (du)$$
with $z_r$ a normalization constant, and  
$$ \mu_0 (du) = \exp\left( -\int_{\R ^d} |\nabla u | ^2 + V |u| ^2 \right)du$$
the free Gibbs (gaussian) measure associated\footnote{I.e., with covariance $(-\Delta + V) ^{-1}$} with $-\Delta +V$. The program of defining and studying the Schr\"odinger flow on the support of $\mu$ has been initiated in~\cite{LebRosSpe-88}, then pursued by many authors and extended to other nonlinear dispersive equations. The first result of measure invariance for a NLS equation is in~\cite{Bourgain-94}.

It is well-known that the free Gibbs measure $\mu_0$ is supported on function spaces of low regularity. This is the main source of difficulty in the definition of the interacting measure $\mu$ and the proof of its invariance under the NLS flow. This is also an important issue as regards the derivation of nonlinear Gibbs measures from many-body quantum mechanics. In~\cite{LewNamRou-14d} we were able to fully control the mean-field limit only when 
\begin{itemize}
 \item [\textbf{(a)}] the gaussian measure is supported at least on $L^2 (\R^d)$;
 \item [\textbf{(b)}] its reduced density matrices are trace-class operators on $L^2 (\R ^d)$;
 \item [\textbf{(c)}] consequently, the construction of the interacting Gibbs measure is straightforward.
\end{itemize}
Essentially this limited us to the 1D case $d=1$ with $-\Delta + V = -\partial_x ^2 + |x| ^s, \: s>2$ (the problem set on a bounded interval is included as the formal case $s=\infty$). In higher dimensions, we were able to derive nonlinear Gibbs measures only for very smooth interaction operators. Multiplication operators by $w(x-y)$ as above, a fortiori by $\delta_0 (x-y)$, were not allowed.

In dimensions $d\geq 2$, properties \textbf{(a)} and \textbf{(b)} fail and a replacement for \textbf{(c)} necessitates a renormalization scheme, a minima a Wick ordering. This has been carried out decades ago in CQFT, see~\cite{GliJaf-87,Simon-74,DerGer-13} for general references. More recently, the corresponding renormalized measures have been shown to be invariant under the (properly renormalized) NLS flow~\cite{Bourgain-96,Bourgain-97,ThoOh-15}. The derivation of these renormalized measures from many-body quantum mechanics is an open problem. The state of the art in this direction is contained in~\cite{FroKnoSchSoh-16} where it has been shown that suitable modifications of bosonic Gibbs states based on renormalized Hamiltonians do converge to the desired measure. Completing the same program for the true Gibbs states remains an important challenge.

\medskip

In this note we address a particular case where 
\begin{itemize}
 \item [\textbf{(d)}] the gaussian measure is \emph{not} supported on $L^2 (\R^d)$;
 \item [\textbf{(e)}] its reduced density matrices are \emph{not} trace-class operators;
 \item [\textbf{(f)}] \emph{nevertheless}, no renormalization is needed to make sense of the interacting measure.
\end{itemize}
In fact, the gaussian measures we shall consider live on some $L^p (\R^d)$, for some $p>2$. That their reduced density matrices are not trace-class has to do with a lack of decay at infinity, rather than a lack of local regularity. 

This situation is somewhat intermediate between the ideal ``trace-class case'', solved in~\cite{LewNamRou-14d}, and the ``Wick renormalized case'', partially solved in~\cite{FroKnoSchSoh-16}. That the 1D harmonic oscillator case $-\Delta + V = -\partial_x ^2 + |x| ^2$ satisfies \textbf{(d)} and \textbf{(f)} has been observed in~\cite{BurThoTzv-10} and used to develop a low-regularity probabilistic Cauchy theory for the 1D nonlinear Schr\"odinger equation. Here we expain that \textbf{(d)} and \textbf{(f)} in fact hold in the case $-\Delta + V = -\partial_x ^2 + |x|^s, s>1 $ and derive the corresponding measures from many-body quantum mechanics. The main point to adapt the strategy of~\cite{LewNamRou-14d} is to overcome the problem posed by \textbf{(e)}. Indeed, the trace-class topology of reduced density matrices (related to moments of the particle number) is the most natural one to pass to the mean-field limit in a many-body quantum problem. The main addition of the present paper is that we are able to work in weaker topologies (namely, the Hilbert-Schmidt and local trace class topologies), to pass to the limit and complete the program of~\cite{LewNamRou-14d}.

\medskip

\noindent\textbf{Acknowledgments:} We received financial support from the French ANR project ANR-13-JS01-0005-01 (N.~Rougerie).

\section{Main result}

We consider the $N$-body quantum Hamiltonian
\begin{equation}\label{eq:ham many}
H_N = \sum_{j=1} ^N h_j + \lambda \sum_{1\leq i < j \leq N} w (x_i-x_j) 
\end{equation}
acting on 
$$\gH _N = L_{\rm sym}^2 (\R^N) \simeq \bigotimes ^N _{\rm sym} L^2 (\R) = \bigotimes ^N _{\rm sym} \gH,$$
the Hilbert space for $N$ bosons\footnote{The assumption of bosonic symmetry is essential. Without it, the mean-field limit of Gibbs states is very different~\cite[Section~3]{LewNamRou-ICMP}.} on the real line, with the symmetric tensor product
$$
f_1\otimes_{\rm sym}\cdots\otimes_{\rm sym} f_N=\frac{1}{\sqrt{N!}}\sum_{\sigma \in S_N}f_{\sigma(1)}\otimes\cdots \otimes f_{\sigma(N)}, \quad \forall f_1,...,f_N\in \gH.
$$
In the above $h_j$ stands for $h$ acting on variable $j$, where 
\begin{equation}\label{eq:harm osc}
h = - \partial_x ^2 + V (x) 
\end{equation}
with a potential $V$ satisfying 
\begin{equation}\label{eq:asum V}
V (x) \geq C^{-1} |x| ^s, \quad s > 1, \quad C >0 . 
\end{equation}
We assume that the interaction potential $w$ is repulsive (defocusing) and decays fast enough at infinity:
\begin{equation}\label{eq:asum w}
0\leq w = w_1+w_2,\qquad w_1\in\cM,\qquad w_2\in L^p(\R) \text{ with } 1\leq p<\frac{1}{(2-s)_+},
\end{equation}
where $\cM$ is the set of bounded (Radon) measures. It is well-known~\cite{ReeSim2} that, under these assumptions, $H_N$ makes sense as a self-adjoint operator on $L^2 (\R ^N)$. The measure part $w_1$ can include a delta function, which is relatively form-bounded with respect to the Laplacian because of the Sobolev embedding. The coupling constant $\lambda \geq 0$ will be scaled appropriately in dependence of the particle number $N$ to make the interaction sufficiently weak for the mean-field approximation to become asymptotically exact. 

Our starting point is the grand-canonical Gibbs state at temperature $T>0$
\begin{equation}\label{eq:Gibbs state}
\Gamma_{\lambda,T} := \frac{\exp\left( - T ^{-1} \bH_\lambda \right)}{\tr_\gF \left[ \exp\left( - T ^{-1} \bH_\lambda \right)\right]}
\end{equation}
where $\bH_\lambda$ is the second quantized version of~\eqref{eq:ham many}:
\begin{equation}\label{eq:ham fock}
\bH_\lambda = \bigoplus_{N=0} ^\infty H_N 
\end{equation}
acting on the bosonic Fock space 
\begin{align}\label{eq:Fock}
\gF &= \C \oplus \gH \oplus \gH_2 \oplus \ldots \oplus \gH_N \oplus \ldots  \nonumber\\
 &= \C \oplus L^2 (\R) \oplus L_{\rm sym}^2 (\R ^2) \oplus \ldots \oplus L^2_{\rm sym} (\R ^N) \oplus \ldots  
\end{align}
The Gibbs state is the unique minimizer over mixed grand canonical states (self-adjoint positive operators on $\gF$ having trace $1$) of the free energy functional
\begin{equation}\label{eq:free ener}
\cF_{\lambda,T} [\Gamma] = \tr_{\gF} \left[ \bH_\lambda \Gamma \right] + T \tr_{\gF} \left[ \Gamma \log \Gamma \right] 
\end{equation}
and the minimum equals 
$$ F_{\lambda,T} = - T \log Z_{\lambda,T}, \quad Z_{\lambda,T} = \tr_\gF \left[ \exp\left( - T ^{-1} \bH_\lambda \right)\right].$$
The method of~\cite{LewNamRou-14d} that we adapt here is variational, based on this minimization principle. To see that $\Gamma_{\lambda,T}$ is indeed the unique solution, observe that for any other state $\Gamma$
\begin{equation}\label{eq:min principle}
 \cF_{\lambda,T} [\Gamma] = \cF_{\lambda,T} \left[\Gamma_{\lambda,T} \right] + T \tr_{\gF} \left[ \Gamma\left( \log \Gamma - \log \Gamma_{\lambda,T} \right)\right]. 
\end{equation}
The last quantity in the right-hand side is the von Neumann relative entropy. It is positive, and equals zero if and only if $\Gamma = \Gamma_{\lambda,T}$, see e.g.~\cite{OhyPet-93,Wehrl-78}. 

We are going to consider the mean-field limit: $T\to \infty$ (corresponding roughly to a large particle number limit) and 
$$\lambda=T^{-1}.$$
The objects that will have a natural limit for large $T$ are the reduced density matrices $\Gamma_{\lambda,T} ^{(k)}$, i.e. the operators on the $k$-particles space $\gH _k$ defined by
\begin{equation}
\Gamma_{\lambda,T}^{(k)}=\sum_{n\geq k}{n \choose k} \tr_{k+1\to n} \left[G_{\lambda,T}^n\right].
\label{eq:def_DM_partial}
\end{equation}
Here $G_{\lambda,T} ^n $ is the projection of $\Gamma_{\lambda,T}$ on the $n$-particle sector $\gH _n$ and $\tr_{k+1\to n}$ is the partial trace taken over the symmetric space of $n-k-1$ variables. Equivalently, we have
\begin{equation}\label{eq:GC DM}
 \tr_{\gH _k} \left[ A_k \Gamma_{\lambda,T}^{(k)}\right]=\sum_{n\geq k} {n\choose k} \tr_{\gH _n} \left[A_k \otimes_{\rm sym} \1 ^{\otimes (n-k)}\,G_{\lambda,T} ^n \right]
\end{equation}
for every bounded operator $A_k$ on $\gH_k$, where 
\begin{equation}
A_k\otimes_{\rm sym}\1^{\otimes{n-k}} = {n\choose k}^{-1}\sum_{1\leq  i_1< \cdots< i_k\leq n}(A_k)_{i_1...i_k}
 \label{eq:A_k}
\end{equation}
and $(A_k)_{i_1,...,i_k}$ acts on the $i_1,..,i_k$-th variables.

The limiting object is the nonlinear Gibbs measure 
\begin{align}\label{eq:NL Gibbs}
d\mu (u) = \frac{1}{z_r} \exp\left( - F_{\rm NL} [u] \right) d\mu_0 (u)
\end{align}
with the nonlinear interaction term 
\begin{align*}
F_{\rm NL} [u] = \frac{1}{2}\iint_{\R \times \R} |u(x)| ^2 w (x-y) |u(y)| ^2 dx dy,
\end{align*}
the relative partition function
\begin{align*}
z_r = \int \exp\left( - F_{\rm NL} [u] \right) d\mu_0 (u),
\end{align*}
and the gaussian measure $\mu_0$ associated with $h$. We refer to Section~\ref{sec:measures} for details, the main points being that 
\begin{itemize}
\item $\mu_0$ can be defined as a measure over $\bigcap_{t<1/2-1/s}\, \gH^{t}$, where $\gH^t$ is the Sobolev-like space 
\begin{equation}\label{eq:Sobolev}
\gH ^t := \left\lbrace u  = \sum_{n= 0} ^\infty \alpha_n u_n \ \big| \ \sum_{n=0} ^\infty \lambda_n ^t |\alpha_n| ^2 < \infty \right\rbrace 
\end{equation}
for $t\in \R$, and the spectral decomposition of $h$ reads\footnote{Using Dirac's bra-ket notation $|u_n\rangle \langle u_n|$ for the orthogonal projector onto $u_n$.} 
\begin{equation}\label{eq:harm osc spec}
h = \sum_{n=0} ^\infty \lambda_n |u_n \rangle \langle u_n | 
\end{equation}
\item $u\mapsto F_{\rm NL} [u]$ is finite $\mu_0$-almost surely, so that $\mu$ is well-defined as a probability measure.
\end{itemize}

To state our main result, we recall a convenient convention from~\cite{LewNamRou-14d}, namely that, for a one-body operator $A$ on $\gH$, we denote $A^{\otimes n}$ the operator on $\gH _n = \bigotimes_{\rm sym}^n \gH$ acting as 
$$ A ^{\otimes n} \left( \varphi_1 \otimes_{\rm sym} \otimes \ldots \otimes_{\rm sym} \varphi_ n \right) =  A \varphi_1 \otimes_{\rm sym} \otimes \ldots \otimes_{\rm sym} A \varphi_ n.$$

The goal of this note is to prove the following:

\begin{theorem}[\textbf{Derivation of Gibbs measures based on (an)harmonic oscillators}]\label{thm:main}
Let $\lambda = T ^{-1}$ and $T\to \infty$. Then, we have the convergence of the relative partition function
\begin{align} \label{eq:rel-partition}
\frac{Z_{\lambda,T}}{Z_{0,T}} = \frac{\tr_\gF \left[ \exp\left( - T ^{-1} \bH_\lambda \right)\right]}{\tr_\gF \left[ \exp\left( - T ^{-1} \bH_0 \right)\right]}  \to z_r >0.
\end{align}
Moreover, for any $k\geq 1$,  
\begin{equation}\label{eq:main DM}
\frac{k!}{T ^k} \Gamma ^{(k)}_{\lambda,T} \to \int |u ^{\otimes k} \rangle \langle u ^{\otimes k} | d\mu (u) 
\end{equation}
in the Hilbert-Schmidt norm, namely 
$$ \tr \left| \frac{k!}{T ^k} \Gamma ^{(k)}_{\lambda,T} - \int |u ^{\otimes k} \rangle \langle u ^{\otimes k} | d\mu (u) 
 \right| ^2 \to 0. $$
\end{theorem}

Note that the limiting measure $\mu$ is uniquely characterized by the collection of the right-hand sides of~\eqref{eq:main DM} for all $k\in\N$.  Before turning to the proof, we make a few comments:

\begin{remark}[Comparison with the trace-class case]\label{rem:main res}\mbox{}\\
In~\cite[Section~5.1]{LewNamRou-14d} we had already proved this result in the case where Assumption~\eqref{eq:asum V} is strengthened to $V(x) \geq C^{-1} |x| ^s, s>2 $. Then, the convergence~\eqref{eq:main DM} is in fact strong in the trace-class and the proof is simpler, for this topology is more easily related to the many-body problem.


In the case under consideration here, the right-hand side of~\eqref{eq:main DM} in fact belongs to the Schatten\footnote{I.e. the sequence of its eigenvalues belongs to $\ell ^p (\N)$, see~\cite{Simon-79}.} class $\gS ^p (\gH_k)$  for any $p> 1/s + 1/2$. The cases $p=1$ and $p=2$ correspond to the trace-class and the Hilbert-Schmidt class, respectively. We \emph{conjecture} that the convergence~\eqref{eq:main DM} is in fact strong in any $\gS ^p (\gH_k)$ with $p> 1/s + 1/2$. 



Note finally that, if $V$ does not increase faster than $|x| ^2$ at infinity, the expected particle number of the grand-canonical Gibbs state has to grow much faster than $T$ in the limit $T\to \infty$. It is then not obvious that choosing $\lambda = T ^{-1}$ should lead to a well-defined mean-field limit, but we prove it does. \hfill$\diamond$
\end{remark}

\section{Gibbs measures based on NLS functionals}\label{sec:measures}

In this section we briefly recall how to construct the interacting Gibbs measure $\mu$.  This has been done for $s>2$ in~\cite{LewNamRou-14d}. The case $s=2$ is covered by~\cite{BurThoTzv-10} (alternative constructions can be based on estimates for Hermite eigenfunctions from e.g.~\cite{KocTat-05,KocTatZwo-07,YajZha-01}). Here we give a softer argument allowing to define the defocusing measure for any $s>1$, without resorting to local smoothing estimates or eigenfunction bounds. 

\medskip

We start with well-known facts on the gaussian measure $\mu_0$. 

\begin{proposition}[\textbf{Free Gibbs measure: definition}]\label{lem:free}\mbox{}\\
Let $h$ be as in~\eqref{eq:harm osc} with $V$ satisfying~\eqref{eq:asum V}. Recall the spectral decomposition~\eqref{eq:harm osc spec}. Define a probability measure $\mu_{0,K}$ on $V_K=\mathrm{span} (u_0,\ldots,u_K)$ by setting
$$d\mu_0 ^K (u) := \bigotimes_{j=0} ^K \frac{\lambda_j}{\pi} \exp\left( -\lambda_j |\langle u, u_j\rangle | ^2\right) d \langle u, u_j\rangle$$
where $d\langle u, u_j\rangle = da_j db_j$ and $a_j,b_j$ are the real and imaginary parts of the scalar product. 

There exists a unique probability measure $\mu_0$ over the space $\bigcap_{t<1/2-1/s} \, \gH^{t}$ such that the measure $\mu_{0,K}$ is the cylindrical projection of $\mu_0$ on $V_K$ for all $K \ge 1$. The corresponding $k$-particle density matrix
\begin{equation}\label{eq:DM free meas}
\gamma_0^{(k)}:=\int |u^{\otimes k}\rangle\langle u^{\otimes k}|\;d\mu_0(u) = k!\,(h^{-1})^{\otimes k}
\end{equation}
belongs to $\gS ^p (\gH _k)$ for all ${1}/{s} + {1}/{2} < p\leq \infty$. 
\end{proposition}
   
\begin{proof} By~\cite[Lemma 1]{Skorokhod-74}, the sequence $\{\mu_{0,K}\}_{K\ge 1}$ defines a unique measure $\mu_0$ on $\gH^t$ if the tightness condition
\begin{equation}
\lim_{R\to\ii}\sup_{K}\mu_{0,K}\big(\{u\in V_K\ :\ \|u\|_{\gH^t}\geq R\}\big)=0
\label{eq:tightness_mu_0}
\end{equation}
holds true. This is satisfied if $\Tr (h^{t-1})<\infty$ since 
$$
\mu_{0,K}\big(\{u\in V_K\ :\ \|u\|_{ \gH^{t}}\geq R\}\big)\leq R^{-2}\int_{V_K}\norm{u}_{\gH^{t}}^2\,d\mu_{0,K}(u)=R^{-2}\sum_{j=1}^K \lambda_j^{t-1}\le R^{-2} \Tr [h^{t-1}].
$$
Applying the Lieb-Thirring inequality in ~\cite[Theorem 1]{DolFelLosPat-06} to  $h=-\Delta+V(x)$, we have
$$\tr h^{-p}\leq 2^p\tr(h+\lambda_0)^{-p}\leq 2^p \int_{\R}\int_{\R} \frac{dx\,dk}{\big(|2\pi k|^2+V(x)+\lambda_0\big)^p}$$
where $\lambda_0>0$ is the lowest eigenvalue of $h$.  Using $V(x)\ge C ^{-1}|x|^s$, we conclude that
\begin{equation}\label{eq:trace cond}
\tr \left[ h ^{-p}\right] < \infty \mbox{ for all } p> 1/s + 1/2. 
\end{equation}
Thus~\eqref{eq:tightness_mu_0} holds true for all $ t < 1/2 - 1/s$, and hence $\mu_0$ is well-defined (uniquely) over $\bigcap_{t < 1/2 - 1/s}\, \gH^{t}$. The formula \eqref{eq:DM free meas} follows from a direct calculation:  
\begin{align*}
&\int |u^{\otimes k}\rangle\langle u^{\otimes k}|\;d\mu_0(u)\nn\\
&\qquad =k!\sum_{i_1\leq i_2\leq\cdots\leq i_k}\left(\prod_{\ell=1}^k\frac{1}{\lambda_{i_\ell}}\right)\frac{|u_{i_1}\otimes_s \cdots \otimes_s u_{i_k}\rangle\langle u_{i_1}\otimes_s \cdots \otimes_s u_{i_k}|}{\norm{u_{i_1}\otimes_s \cdots \otimes_s u_{i_k}}^2} =k!\,(h^{-1})^{\otimes k},
\end{align*}
see \cite[Lemma 3.3]{LewNamRou-14d}  for details. 
\end{proof}

In order to make sense of the interacting measure, we need to prove that the gaussian measure is in fact supported on $L^p$ spaces. 

\begin{lemma}[\textbf{Free Gibbs measure: support}]\label{lem:free supp}\mbox{}\\
The gaussian measure $\mu_0$ constructed in Proposition~\ref{lem:free}  is supported on $L^{r}(\R)$ for every 
$$\max(2,4/s)<r<\infty.$$ 
More precisely, there exists $\alpha_r > 0$ such that
\begin{equation}
\int e^{\alpha_r\norm{u}_{L^{r}(\R)}^2}d\mu_0(u)<\ii.
\label{eq:mu_0_L_p}
\end{equation}
\end{lemma}

\begin{proof}
Consider the kernel of the operator $h^{-1}$ (the eigenfunctions $u_n$ can be chosen real-valued)
\begin{equation}\label{eq:exp dens}
h^{-1}(x;y) = \sum_{n\geq 0} \frac{1}{\lambda_n}  u_n (x) u_n (y).
\end{equation}
Note that $h^{-1}(x;x) \ge 0$. 

\medskip

\noindent \textbf{Step 1. } We claim that $x\mapsto h^{-1}(x;x)$ belongs to $L^p(\R)$ for all 
$$\max(1,2/s)<p\leq\ii.$$
We will prove that, for any function/multiplication operator $\chi \ge 0$ satisfying $\chi^2\in L^{q}(\R)$ with $1/p+1/q=1$, the operator $\chi h^{-1}\chi$ is trace class and
\begin{equation} \label{eq:local-trace}
 \tr \left[ \chi h ^{-1} \chi \right]  = \norm{h^{-1/2}\chi}_{\gS^2(\gH)} ^2 \leq C \norm{\chi ^2}_{L^q(\R)}.
 \end{equation}
Let us estimate the Hilbert-Schmidt norm of $h ^{-1/2} \chi$.
%
We pick some $0<\alpha < 1/2 $, write 
\begin{equation}
h^{-1/2}\chi=h^{\alpha-1/2} \left( h^{-\alpha}(1-\partial_x ^2 )^{\alpha} \right) \left((1-\partial_x ^2)^{-\alpha}\chi \right)
\label{eq:decomp_h_chi}
\end{equation}
and estimate the three factors separately.  First, returning to~\eqref{eq:trace cond} we have 
\begin{equation}\label{eq:trace-h-eta}
h^{\alpha-1/2}\in \gS^{2p}(\gH)\qquad\text{for}\quad 2p\left(\frac{1}{2}-\alpha\right)>\frac{1}{s}+\frac{1}{2}.
\end{equation}
Second, $h \geq C^{-1} (1-\partial_x ^2)$ as operators, for some constant $C>0$. Indeed 
$$ h = -\partial_x ^2 + V  \geq \frac{1}{2} \left( -\partial_x ^2 + \lambda_0 \right)$$
with $\lambda_0 >0$ the lowest eigenvalue of $h$. Thus, using the operator-monotonicity~\cite[Theorem~V.1.9]{Bhatia} of $x\mapsto x^{2\alpha}$ for $0<\alpha\leq 1/2$, we deduce that 
$$ h^{2\alpha} \geq C^{-2\alpha} (1-\partial_x ^2)^{2\alpha},$$ and thus 
\begin{equation}\label{eq:bounded alpha}
 h^{-\alpha} (1-\partial_x ^2)^{2\alpha}h^{-\alpha}\leq C^{2\alpha}.
\end{equation}
In particular, $ h^{-\alpha} (1-\partial_x ^2)^{\alpha}$ is a bounded operator for every $\alpha\leq1/2$.

Third, we aply the Kato-Seiler-Simon inequality~\cite[Theorem 4.1]{Simon-79} to get
\begin{equation}\label{eq:KSS}
\norm{(1-\partial_x ^2)^{-\alpha}\chi}_{\gS^{2q}(\gH)}\leq\left(\int_{\R}\frac{dk}{(1+|2\pi k|^2)^{2\alpha q}}\right)^{\frac{1}{2q}}\norm{\chi}_{L^{2q}(\R)}
\end{equation}
when $q\ge 1$ and $4\alpha q>1$. Combining \eqref{eq:decomp_h_chi} with~\eqref{eq:trace-h-eta},~\eqref{eq:bounded alpha} and~\eqref{eq:KSS} we infer from H\"older's inequality \cite[Theorem 2.8]{Simon-79} that\footnote{$\norm{\, . \,}_{\gS ^\infty}$ stands for the operator norm.}
\begin{equation}\label{eq:proof Lp}
\norm{h^{-1/2}\chi}_{\gS^2(\gH)}\leq \|h^{\alpha-1/2}\|_{\gS^{2p}} \| h^{-\alpha}(1-\partial_x ^2)^{\alpha} \|_{\gS^{\infty}} \| (1-\partial_x ^2)^{-\alpha}\chi\|_{\gS^{2q}} \le C\norm{\chi}_{L^{2q}(\R)}
\end{equation}
for $1/p+1/q=1$. 
The two constraints that $2p(1/2-\alpha)> 1/s+1/2$ and $4\alpha q > 1$ require
$$\frac{1}{2}= \left( \frac{1}{2} - \alpha \right)+ \alpha >\frac{1}{2p} \left( \frac{1}{s} +\frac{1}{2} \right) + \frac{1}{4q} = \frac{1}{2ps} +\frac{1}{4},$$
or equivalently
$$p> \frac{2}s.$$
Thus  \eqref{eq:proof Lp}, and hence \eqref{eq:local-trace}, holds true for all $p>\max(1,2/s)$. Note that \eqref{eq:local-trace} implies that $h^{-1}$ is locally trace-class, which ensures that $h^{-1}(x;x) \in L^1_{\rm loc}(\R)$ and 
$$
 \int_{\R} h ^{-1} (x;x) \chi ^2 (x) dx  = \tr \left[ \chi h ^{-1} \chi \right]  =  \norm{h^{-1/2}\chi}_{\gS^2(\gH)} ^2 \le C \norm{\chi ^2}_{L^q(\R)}.
 $$
By duality, we conclude that $x\mapsto h^{-1}(x;x)\in L^p(\R)$ for all $p>\max(1,2/s)$.

\medskip

\noindent\textbf{Step 2.} We deduce from the above that $\mu_0$ is supported on $L^r(\R^d)$ for $r>\max(2,4/s)$. 

We will use an interpolation argument in the spirit of Khintchine's inequality (see, e.g.~\cite[Lemma~4.2]{BurTzv-08}). Formally, when $r=2k$ is an even integer, by considering the diagonal of the kernels of operators in \eqref{eq:DM free meas},  we have
\begin{align}\label{eq:even k}
\int|u(x)| ^{2k}  d\mu_0 (u)  =  k! [h ^{-1} (x;x)] ^k. 
\end{align}
Then by interpolation, we get 
\begin{align*}
\int|u(x)| ^{r}  d\mu_0 (u)  \le  C_r [h ^{-1} (x;x)] ^{\frac{r}{2}}
\end{align*}
for all $r\ge 2$. The right side is integrable when $r>\max(2,4/s)$ by Step 1. 

Now we go to the details with full rigor. Let $P_K$ be the projection onto $V_K={\rm span}(u_0,...,u_K)$. Using 
$$
\int \langle u_j, u\rangle d\mu_0(u) = 0, \quad \int |\langle u_j, u\rangle|^2 d\mu_0(u) = \lambda_j^{-1}
$$
we obtain
$$
\int | P_Ku(x)|^{2}d\mu_0(u) =  \int \left|\sum_{j=0}^K\pscal{u_j,u}u_j(x)\right|^{2}d\mu_0(u) = \sum_{j=0}^K \frac{|u_j(x)|^2}{\lambda_j} \le h^{-1}(x;x).
$$
More generally, when $r=2k$ is an even integer ($k=1,2,3,...$), by Wick's theorem we can compute
\begin{align} \label{eq:Khintchine}
\left(\int | P_Ku(x)|^{r}d\mu_0(u) \right)^{\frac{2}{r}} & =  \left(\int \left|\sum_{j=0}^K\pscal{u_j,u}u_j(x)\right|^{2k}d\mu_0(u) \right)^{\frac{1}{k}} \nn\\
& \le C_r \sum_{j=0}^K\frac{|u_j(x)|^2}{\lambda_j}  \le C_r h^{-1}(x;x).
\end{align}
By H\"older's inequality in $L^p$ spaces associated with the measure $\mu_0$, we can extend \eqref{eq:Khintchine} to all $r\ge 2$. Then we rewrite this inequality as
$$
\int | P_Ku(x)|^{r}d\mu_0(u) \le C_r [h^{-1}(x;x)]^{\frac{r}{2}}
$$
and integrate over $x\in \R$. This gives
\begin{equation*}
\int\norm{P_Ku}_{L^{r}(\R)}^{r}d\mu_0(u)\leq C_r\int_\R [h^{-1}(x;x)]^{\frac{r}{2}}\,dx
\end{equation*}
where the right side is finite for $r>\max(2,4/s)$. 
Passing to the limit $K\to\ii$, we find that $\norm{u}_{L^{r}(\R)}$ is finite $\mu_0$-almost surely and
\begin{equation*}
\int\norm{u}_{L^{r}(\R)}^{r}d\mu_0(u)\leq C_r\int_\R [h^{-1}(x;x)]^{\frac{r}{2}}\,dx.
\end{equation*}
Then, by Fernique's theorem~\cite{Fernique-70}, there must exist a number $\alpha_r>0$ such that~\eqref{eq:mu_0_L_p} holds. 
\end{proof}

As regards the interacting measure we deduce the following.

\begin{corollary}[\textbf{Interacting Gibbs measure}]\label{cor:NL gibbs}\mbox{}\\
Let $h$ be as in~\eqref{eq:harm osc} with $V$ satisfying~\eqref{eq:asum V} and $w$ be as in~\eqref{eq:asum w}. Then the functional
$$ u \mapsto F_{\rm NL} [u] = \frac{1}{2} \iint_{\R \times \R} |u(x)| ^2 w (x-y) |u(y)| ^2 dx dy \geq 0$$
is in $L^1 (d\mu_0)$,
$$
\int F_{\rm NL} [u] d\mu_0 (u) <\infty.
$$
In particular, $F_{\rm NL} [u]$ is finite $\mu_0$-almost surely. Thus, the measure defined by~\eqref{eq:NL Gibbs} makes sense as a probability measure on $\bigcap_{t<1/2-1/s} \, \gH ^t$ and 
$$z_r = \int \exp\left( - F_{\rm NL} [u] \right) d\mu_0 (u) >0.$$
\end{corollary}

\begin{proof}
Since $w\geq0$ we have $F_{\rm NL } [u] \geq 0$ and it is sufficient to show that its integral with respect to $\mu_0$ is finite. Writing $w=w_1+w_2$ as in~\eqref{eq:asum w}, this follows immediately from~\eqref{eq:mu_0_L_p} since
$$F_{\rm NL } [u]\leq \norm{w_1}_{\mathcal{M}}\norm{u}_{L^4(\R)}^4+\norm{w_2}_{L^p(\R)}\norm{u}_{L^r(\R)}^4$$
by Young's inequality, with $4/r+1/p=2$. 
\end{proof}

\section{Hilbert-Schmidt estimate} \label{sec:kernel}

We shall henceforth denote points in $\R^k$ in the manner $X_k = (x_1,\ldots,x_k)$ and denote $dX_k$ the corresponding Lebesgue measure. Very often we identify a Hilbert-Schmidt operator $A_k$ on $L ^2 (\R ^k)$ with its integral kernel $A_k (X_k;Y_K)$
$$ (A_k \Psi_k) (X_k) = \int_{\R ^k} A_k (X_k;Y_k) \Psi_k (Y_k) dY_k.$$

The main new estimate we need to put the proof strategy of~\cite{LewNamRou-14d} to good use is the following

\begin{proposition}[\textbf{Bounds in Hilbert-Schmidt norm}]\label{pro:new bounds}\mbox{}\\
Let the reduced density matrices $\Gamma_{\lambda,T}^{(k)}$ be defined as in~\eqref{eq:GC DM}, with $\lambda = T^{-1}$. Then we have the integral kernel estimate
\begin{equation}\label{eq:new kernel bound}
0\le \Gamma_{\lambda,T}^{(k)} (X_k;Y_k) \le C\, \Gamma_{0,T}^{(k)} (X_k;Y_k).
\end{equation}
Consequently
\begin{equation}\label{eq:new bound}
\tr_{\gH _k} \left[ \left( \Gamma_{\lambda,T} ^{(k)}\right) ^2 \right] \leq C^2 \tr_{\gH _k} \left[ \left( \Gamma_{0,T} ^{(k)}\right) ^2\right]   \leq C^2 T ^{2k}\left(\tr(h^{-2})\right)^{k}
\end{equation}
for all $k\in \mathbb{N}$.
\end{proposition}

Note that the density matrices of the non-interacting Gibbs state $\Gamma_{0,T}$ are given by~\cite[Lemma 2.1]{LewNamRou-14d}
\begin{equation} \label{eq:dm-free-Gibbs}
\Gamma_{0,T}^{(k)}= \left(\frac{1}{e^{h/T}-1}\right)^{\otimes k} \le T^k (h^{-1})^{\otimes k}.
\end{equation}
Therefore, the second inequality in~\eqref{eq:new bound} follows immediately from the fact that $h^{-1}\in \gS^2(\gH)$, see Proposition \ref{lem:free}. The first inequality in \eqref{eq:new bound} follows from \eqref{eq:new kernel bound} and the well-known fact that the $L^2$-norm of the kernel is equivalent to the Hilbert-Schmidt norm of the operator, see e.g~\cite[Theorem~VI.23]{ReeSim1}.

It remains to prove \eqref{eq:new kernel bound}. This is very much in the spirit of~\cite[Theorem~6.3.17]{BraRob2}, which is proved using a Feynman-Kac representation of reduced density matrices originating in~\cite{Ginibre-65b,Ginibre-65c,Ginibre-65d} (see also~\cite{FroPar-78,FroPar-80}). We certainly could obtain such a representation, in the spirit of~\cite[Theorem~6.3.14]{BraRob2}. However, we do not need to go that far to obtain the desired bound: the Trotter product formula is sufficient for our purpose. 

Our proof of \eqref{eq:new kernel bound} is based on two useful lemmas. The first is essentially taken from ~\cite[Lemma~8.1]{LewNamRou-14d}.

\begin{lemma}[\textbf{Bounds on partition functions}]\label{lem:partition}\mbox{}\\
Let the partition function be defined as 
\begin{equation}\label{eq:partition}
Z_{\lambda,T} = \tr_\gF \left[ \exp\left( - T ^{-1} \bH_\lambda \right)\right]. 
\end{equation}
Then, for $\lambda = T^{-1}$, we have
\begin{equation}\label{eq:part bound}
1 \leq \frac{Z_{0,T}}{Z_{\lambda,T}} \leq C 
\end{equation}
where the constant $C>0$ is independent of $T$.
\end{lemma}

\begin{proof} Using  $w\ge 0$, we have $\bH_{\lambda}\ge \bH_{0}$, and hence 
$$Z_{\lambda,T}=\tr_\gF \left[ \exp\left( - T ^{-1} \bH_\lambda \right)\right] \le \tr_\gF \left[ \exp\left( - T ^{-1} \bH_0 \right)\right]  =Z_{0,T}.$$
On the other hand, since $\Gamma_{\lambda,T}$ minimizes the free energy functional $\cF_{\lambda,T}(\Gamma)$ in \eqref{eq:free ener}, 
$$
-T\log Z_{\lambda,T}= \cF_{\lambda,T}(\Gamma_{\lambda,T}) \le \cF_{\lambda,T}(\Gamma_{0,T})=-T\log Z_{0,T} + \lambda \Tr[w \Gamma_{0,T}^{(2)}].
$$
Inserting \eqref{eq:dm-free-Gibbs} and $\lambda=T^{-1}$ into the latter estimate, we conclude that  
$$
-\log \frac{Z_{\lambda,T}}{Z_{0,T}} \le \lambda T^{-1} \Tr[w \Gamma_{0,T}^{(2)}] \le \tr [w \, h ^{-1} \otimes h ^{-1}] <\infty. 
$$
Here the last estimate is taken from Corollary~\ref{cor:NL gibbs}. 
\end{proof}

The second lemma is a well-known comparison result for the heat kernels of Schr\"odinger operators on $L^2(\R^n)$ (with no symmetrization).

\begin{lemma}[\textbf{Heat kernel estimate}]\label{lem:int ker bound}\mbox{}\\
Consider two Schr\"odinger operators $K_j= -\Delta_{\R^n}+ W_j$ on $L^2(\R^n)$, $j=1,2$,  with $W_1 \ge W_2\ge 0$. Then for all $t>0$, we have the integral kernel estimate  
\begin{equation}\label{eq:int ker bound}
0\leq \exp(-tK_1)(X_n;Y_n) \leq \exp(-t K_2) (X_n;Y_n) 
\end{equation}
for almost every $(X_n;Y_n)\in \R^{n} \times \R ^n$.
\end{lemma}

\begin{proof}
This follows e.g. from the considerations of ~\cite[Sec.~II.6]{Simon-05}. According to the Trotter product formula (see e.g.~\cite[Theorem VIII.30]{ReeSim1} or~\cite[Theorem~1.1]{Simon-05}), we have, for any $\Psi_n,\Phi_n \in L ^2 (\R^n)$,
$$ 
\left\langle \Psi_n | \exp(-tK_j) | \Phi_n \right\rangle = \lim_{m\to \infty} \left\langle \Psi_n \big | \left( \exp \left( \frac{t\Delta_{\R^n}}{m} \right)  \exp\left( -\frac{t W_j}{m}  \right) \right) ^m \big| \Phi_n \right\rangle.
$$
In terms of integral kernels this means 
\begin{align*} 
&\int \overline{\Psi_n (X_n)} \exp(-tK_j) (X_n;Y_n) \Phi_n (Y_n) dX_n d Y_n \\
&\qquad = \lim_{m\to \infty} \int \overline{\Psi_n (X_n)} \exp \left( \frac{t\Delta_{\R^n}}{m} \right)  (X_n;Z_n ^1) \exp\left(-\frac{tW_j(Z_n^1)}{m}   \right) \ldots \\ 
&\qquad\qquad\qquad \exp \left( \frac{t\Delta_{\R^n}}{m} \right) (Z_n ^{m-1};Y_n)  \exp\left( -\frac{tW_j(Y_n)}{m}\right) \Phi_n (Y_n) dX_n dZ_n ^1 \ldots dZ_n ^{m-1} d Y_n
\end{align*}
where the $Z_n ^k = (z_1 ^k,\ldots,z_n ^k)$ are auxiliary sets of variables in $\R^n$ that we integrate over. Therefore, we can specialize to nonnegative functions $\Psi_n, \Phi_n$ and obtain
\begin{align} \label{eq:lower-upper-Trotter}
0 &\le \int \Psi_n (X_n) \exp(-tK_1) (X_n;Y_n) \Phi_n (Y_n) dX_n d Y_n \nn\\
&\qquad \qquad \qquad \le \int \Psi_n (X_n) \exp(-tK_2) (X_n;Y_n) \Phi_n (Y_n) dX_n d Y_n.
\end{align}
Here we have used the fact that the heat kernel $\exp(\frac{t}{m}\Delta_{\R^n})(X_n;Y_n)$ is positive and
$$ 0\le \exp\left(-\frac{tW_1}{m}\right) \le \exp\left(-\frac{tW_2}{m}\right)\quad \text{pointwise.}$$
There remains to let $\Psi_n,\Phi_n$ converge to delta functions in \eqref{eq:lower-upper-Trotter} to conclude the proof.
\end{proof}

Now we can give the 

\begin{proof}[Proof of Proposition~\ref{pro:new bounds}] Our bosonic state $\Gamma_{\lambda,T}$ can be written in the unsymmetrized Fock space in the manner 
\begin{equation}\label{eq:without sym}
\Gamma_{\lambda,T} = \frac{1}{Z_{\lambda,T}} \bigoplus_{n=0} ^{\infty} P_{\rm sym} ^n \exp\left( - T ^{-1} H_n \right)  
\end{equation}
with the symmetric projector
$$ P_{\rm sym} ^n = \frac{1}{n!} \sum_{\sigma \in S_n} U_{\sigma}.$$
Here the sum is over the permutation group $S_n$ and $U_{\sigma}$ is the unitary operator permuting variables according to $\sigma$. We consider $P_{\rm sym} ^n \exp\left( - T ^{-1} H_n \right)$ as an operator on $L^2 (\R ^n)$. Note that $P_{\rm sym} ^n$ commutes with $H_n$ and, in terms of integral kernels, 
$$
\left[ P_{\rm sym} ^n \exp\left( - T ^{-1} H_n \right) \right](X_n;Y_n) = \frac{1}{n!}\sum_{\sigma\in S_n}  \left[ \exp\left( - T ^{-1} H_n \right) \right](\sigma\cdot X_n;Y_n)
$$
where $\sigma\cdot X_n=(x_{\sigma(1)},...,x_{\sigma(n)})$ are the permuted variables. 

By applying Lemma~\ref{lem:int ker bound} to the potentials
$$W_1(X_n)=\sum_{j=1}^n V(x_j)+\lambda \sum_{1\le i<j\le n} w(x_i-x_j) \ge \sum_{j=1}^n V(x_j) =W_2(X_n)$$
(as $w\ge 0$) we have 
$$
\left[ \exp\left( - T ^{-1} H_n \right) \right](X_n;Y_n) \le \left[ \exp\left( - T ^{-1} \sum_{j=1}^n h_j \right) \right](X_n;Y_n).
$$
Since the kernel estimate remains unchanged by the symmetrization\footnote{Which would {\em not} be true if we were dealing with fermions, i.e. $P_{\rm sym} ^n$ was replaced by the antisymmetric projector.}, we have
\begin{align} \label{eq:int ker bound bis}
\left[ P_{\rm sym} ^n \exp\left( - T ^{-1} H_n \right) \right](X_n;Y_n) \le \left[ P_{\rm sym} ^n \exp\left( - T ^{-1} \sum_{j=1}^n h_j \right) \right](X_n;Y_n).
\end{align}
Finally, by Definition~\eqref{eq:def_DM_partial}, the integral kernel of $\Gamma_{\lambda,T} ^{(k)}$ is given by 
$$ 
\Gamma_{\lambda,T} ^{(k)}(X_k;Y_k)= \frac{1}{Z_{\lambda,T}} \sum_{n\geq k} {n \choose k} \int \left[ P_{\rm sym} ^n \exp\left( - T ^{-1} H_n \right) \right] (X_k,Z_{n-k};Y_k,Z_{n-k}) dZ_{n-k} 
$$
with $Z_{n-k} = (z_{k+1},\ldots,z_n) \in \R ^{n-k}$. Inserting~\eqref{eq:int ker bound} into the latter formula, we thus obtain   
$$ 
0\le \Gamma_{\lambda,T} ^{(k)}(X_k;Y_k) \le \frac{Z_{0,T}}{Z_{\lambda,T}} \Gamma_{0,T} ^{(k)}(X_k;Y_k) \le C \Gamma_{0,T} ^{(k)}(X_k;Y_k).
$$
Here the last estimate follows from Lemma~\ref{lem:partition}. 
\end{proof}

\section{Proof of the main theorem}

As in~\cite{LewNamRou-14d}, our strategy is based on Gibbs' variational principle, which states that $\Gamma_{\lambda,T}$ minimizes the free energy functional $\cF_{\lambda,T}[\Gamma]$ in \eqref{eq:free ener}. It follows from a simple computation that $\Gamma_{\lambda,T}$ is also the unique minimizer for the {\em relative free energy functional}:
\begin{align} \label{eq:rel-energy}
-\log \frac{Z_{\lambda,T}}{Z_{0,T}}&=\frac{\cF_{\lambda,T}(\Gamma_{\lambda,T}) - \cF_{0,T}(\Gamma_{0,T})}{T} \nn\\
&= \inf_{\substack{\Gamma\geq 0\\ \tr_{\gF}\Gamma=1}} \Big( \cH(\Gamma,\Gamma_{0,T}) + T^{-2} \Tr[w \Gamma_{\lambda,T}] \Big).
\end{align}
Here  
$$\cH(\Gamma,\Gamma')= \tr_{\gF}\big(\Gamma(\log\Gamma-\log\Gamma')\big) \ge 0$$
is called the \emph{relative entropy}~\cite{OhyPet-93,Wehrl-78} of two states $\Gamma$ and $\Gamma'$. 

We will relate the quantum problem \eqref{eq:rel-energy} to  its classical version: The interacting Gibbs measure $\mu$ is the unique minimizer for the variational problem
\begin{equation} \label{eq:zr-rel}
-\log z_r = \inf_{\nu \text{~probability measure}} \Big( \cH_{\rm cl}(\nu, \mu_0) + \frac{1}{2} \int \langle u^{\otimes 2}, w  u^{\otimes 2} \rangle \,d\nu(u)\Big)
\end{equation}
where\footnote{Positivity of this quantity follows from Jensen's inequality. It is zero if and only if $\nu = \nu'$.} 
$$ \cH_{\rm cl} (\nu,\nu'):= \int_{\gH^{s}}\frac{d\nu}{d\nu'}(u)\log\left(\frac{d\nu}{d\nu'}(u)\right)\,d\nu'(u) \ge 0 $$
is the classical relative entropy of two probability measures $\nu$ and $\nu'$. 

Note that $\cH_{\rm cl}(\nu, \mu_0) = +\infty$ unless $\nu$ is absolutely continuous with respect to $\mu_0$, and the other term of the functional is positive. Thus the minimization above is amongst measures of the form $d\nu (u) = f(u) d\mu_0 (u)$ that all live over $L^4 (\R)$ as per Lemma~\ref{lem:free supp}. Hence the variational problem makes sense. To see that $\mu$ is  the unique minimizer, one argues exactly as in~\eqref{eq:min principle}. 

\subsection{Convergence of the relative partition function} Let us prove \eqref{eq:rel-partition}. We recall the following result from \cite[Lemma 8.3]{LewNamRou-14d}.

\begin{lemma}[\textbf{Free-energy upper bound} ]\label{lem:upper_bound}\mbox{}\\
Let $h>0$ satisfy $h^{-1}\in\gS^p (\gH)$ for some $1\le p<\infty$ and let $w\geq0$ satisfy  
$$\tr_{\gH_2}\Big[w\,h^{-1}\otimes h^{-1}\Big]<\ii.$$ 
Then we have
\begin{align} \label{eq:rel-upper}
\limsup_{T\to \infty} \left( -\log \frac{Z_{\lambda,T}}{Z_{0,T}} \right)\le -\log z_r.
\end{align}
\end{lemma}

Note that 
$$\tr_{\gH_2}\Big[w\,h^{-1}\otimes h^{-1}\Big]=\iint_{\R\times \R} w(x-y)\,h^{-1}(x;x)\,h^{-1}(y;y)\,dx\,dy$$ 
is finite by the proof of Proposition~\ref{lem:free}, under our assumptions on $h$ and $w$. Therefore, the upper bound \ref{eq:rel-upper} holds true. The main difficulty is to establish the matching lower bound. To do this, we need two tools from \cite{LewNamRou-14d}. 

The first one is a variant of the quantum de Finetti Theorem in Fock space~\cite[Theorem~4.2]{LewNamRou-14d} (whose proof goes back to the analysis of~\cite{AmmNie-08,LewNamRou-14}, see~\cite{Rougerie-LMU,Rougerie-cdf} for a general presentation).

\begin{theorem}[\textbf{Quantum de Finetti theorem in Schatten classes}]\label{lem:deFinetti}\mbox{}\\
Let  $\{\Gamma_n\}$ be a sequence of states on the bosonic Fock space $\gF$, namely $\Gamma_n$ is a self-adjoint operator with $\Gamma_n\ge 0$ and $\tr_{\gF}\Gamma_n=1$. Assume that there exists a sequence $\epsilon_n\to0^+$ such that
\begin{equation}
(\epsilon_n)^{pk} \Tr_{\gH_k} \Big[\Big(\Gamma_n^{(k)} \Big)^p \Big] \le C_k\ <\infty, 
\label{eq:bound_DM}
\end{equation}
for some $1\leq p<\ii$ and for all $k\ge 1$. Let $h >0$ be a self-adjoint operator on $\gH$ with 
\begin{equation}\label{eq:h-p again}
 \tr_{\gH} [ h ^{-p}] < \infty 
\end{equation}
and $\gH ^{1-p}$ the associated Sobolev space~\eqref{eq:Sobolev}.

Then, up to a subsequence of $\{\Gamma_n\}$, there exists a Borel probability measure $\nu$ on $\gH^{1-p}$ (invariant under multiplication by a phase factor), called the de Finetti measure of $\Gamma_{n}$ at scale $\eps_{n}$, such that  
\begin{equation}
k!(\epsilon_{n})^k\,\Gamma_{n}^{(k)}\wto \int_{\gH^{1-p}} |u^{\otimes k}\rangle \langle u^{\otimes k}| \,d\nu(u)
\label{eq:weak_limit_DM}
\end{equation}
weakly-$\ast$ in $\gS ^p (\gH _k)$ for every $k\ge 1$.
\end{theorem}

\begin{proof}
This follows straightforwardly from~\cite[Theorem 4.2]{LewNamRou-14d}. Using~\eqref{eq:bound_DM},~\eqref{eq:h-p again} and the H\"older inequality in Schatten spaces, one readily checks that Assumption~(4.7) of~\cite[Theorem 4.2]{LewNamRou-14d} is satisfied for all integer $s$. Convergence of density matrices, along a subsequence, to the right-hand side of~\eqref{eq:weak_limit_DM} in a weaker topology is then Statement~(4.9) of~\cite[Theorem 4.2]{LewNamRou-14d}. Passing to a further subsequence,~\eqref{eq:bound_DM} allows to get weakly-$\ast$ convergence in $\gS ^p (\gH _k)$.
\end{proof}

The second tool is a link between the quantum relative entropy and the classical one, taken from \cite[Theorem 7.1]{LewNamRou-14d} (this is a Berezin-Lieb-type inequality, its proof goes back to the techniques in~\cite{Berezin-72,Lieb-73b,Simon-80}).

\begin{theorem}[\textbf{Relative entropy: quantum to classical}] \label{lem:rel-entropy}\mbox{}\\
Let $\{\Gamma_n\}$  and $\{\Gamma'_n\}$ be two sequences of states on the bosonic Fock space $\gF$. Assume that they satisfy the assumptions of Theorem~\ref{lem:deFinetti} with the same scale $\epsilon_n\to0^+$ and the same power $p\ge 1$. Let $\mu$ and $\mu'$ be the corresponding de Finetti measures. Then 
$$
\liminf_{n\to \infty} \cH (\Gamma_n,\Gamma_n') \ge \cH_{\rm cl} (\mu,\mu').$$
\end{theorem}

Now we are ready to prove a lower bound to the relative free energy matching the upper bound of Lemma~\ref{lem:upper_bound}. 

\begin{lemma}[\textbf{Free-energy lower bound} ]\label{lem:lower_bound}\mbox{}\\
With the notation and assumptions of Theorem~\ref{thm:main} we have
\begin{align} \label{eq:rel-lower}
\liminf_{T\to \infty} \left( -\log \frac{Z_{\lambda,T}}{Z_{0,T}} \right)\ge -\log z_r.
\end{align}
\end{lemma}

\begin{proof}
We pass to the liminf first in the relative entropy and then in the interaction energy.

\medskip

\noindent\textbf{Step 1.} From the Hilbert-Schmidt estimate \eqref{eq:new bound} in Proposition \ref{pro:new bounds}, we can apply Theorem~\ref{lem:deFinetti} to the sequence $\{\Gamma_{\lambda_n,T_n}\}$ for any $T_n \to \infty$, with scale $\eps_n=T_n^{-1}$. Thus, up to a subsequence of $\{\Gamma_{\lambda_n,T_n}\}$, there exists is a Borel probability measure $\nu$ on $\gH^{-1}$ (the de Finetti measure for $\{\Gamma_{\lambda_n,T_n}\}$) such that
\begin{equation}
\frac{k!}{T_n^k}\Gamma_{\lambda_n,T_n}^{(k)}\wto \int_{\gH^{-1}} |u^{\otimes k}\rangle \langle u^{\otimes k}| \,d\nu(u)
\label{eq:weak_limit_DM_n}
\end{equation}
weakly  in $\gS ^2 (\gH _k)$ for every $k\ge 1$. Next, from \eqref{eq:dm-free-Gibbs} and \eqref{eq:DM free meas}, by Lebesgue's Dominated Convergence Theorem we find that
$$
\frac{k!}{T ^k} \Gamma_{0,T}^{(k)} \to \int |u ^{\otimes k} \rangle \langle u ^{\otimes k} | d\mu_0(u) 
$$
strongly in $\gS ^2 (\gH _k)$ for every $k\ge 1$. In particular, the free Gibbs measure $\mu_0$ is the de Finetti measure for the sequence $\{\Gamma_{0,T_n}\}$ with scale $\eps_n=T_n^{-1}$.  Therefore, Lemma~\ref{lem:rel-entropy} implies that
\begin{equation} \label{eq:rel-Gn}
\liminf_{n\to \infty} \cH (\Gamma_{\lambda_n,T_n},\Gamma_{0,T_n}) \ge  \cH_{\rm cl} (\nu,\mu_0). 
\end{equation}
Consequently, $\cH_{\rm cl} (\nu,\mu_0)$ is finite and thus $\nu$ is absolutely continuous with respect to $\mu_0$. In particular, $\nu$ is supported on $L^4(\R)$ by Lemma~\ref{lem:free supp}. 

\medskip

\noindent\textbf{Step 2.} From Lemma~\ref{lem:upper_bound} and the variational principle, it follows that 
$$ T_n^{-2}\Tr \left[w^{1/2} \Gamma_{\lambda_n,T_n} w ^{1/2} \right] \leq C $$
and thus the positive operator $T_n ^{-2}w^{1/2} \Gamma_{\lambda_n,T_n} w ^{1/2}$ has a trace-class weak-$\ast$ limit along a subsequence.  
Using~\eqref{eq:weak_limit_DM_n} with $k=2$ to identify the limit and Fatou's lemma for operators\footnote{Lower semi-continuity of the trace in the weak-$\ast$ topology.}, we get
\begin{equation} \label{eq:rel-wn}
\liminf_{n\to \infty} T_n^{-2}\Tr \left[w \Gamma_{\lambda_n,T_n} \right] \ge \frac{1}{2} \int \langle u^{\otimes 2}, w  \, u^{\otimes 2} \rangle \,d\nu(u).
\end{equation}
Note that on the right side of \eqref{eq:rel-wn}, $\langle u^{\otimes 2}, w  \, u^{\otimes 2} \rangle$ is finite when $u\in L^r(\R)$ for $\max(2,4/s)<r<\ii$. 

Putting \eqref{eq:rel-Gn} and \eqref{eq:rel-wn} together, then combining with \eqref{eq:rel-energy} and \eqref{eq:zr-rel}, we arrive at
\begin{align} \label{eq:rel-lower bis}
\liminf_{n\to \infty} \left( -\log \frac{Z_{\lambda_n,T_n}}{Z_{0,T_n}} \right) &= \liminf_{n\to \infty}  \left( \cH (\Gamma_{\lambda_n,T_n},\Gamma_{0,T_n}) + T_n^{-2}\Tr \Big[w \Gamma_{\lambda_n,T_n} \Big]  \right) \nn\\
&\ge  \cH_{\rm cl} (\nu,\mu_0) + \frac{1}{2} \int \langle u^{\otimes 2}, w  u^{\otimes 2} \rangle \,d\nu(u) \ge -\log z_r.
\end{align}
From~\eqref{eq:rel-lower bis} and the upper bound \eqref{eq:rel-upper}, we conclude that 
\begin{align} \label{eq:rel-equal}
\liminf_{n\to \infty} \left( -\log \frac{Z_{\lambda_n,T_n}}{Z_{0,T_n}} \right) = \cH_{\rm cl} (\nu,\mu_0) + \frac{1}{2} \int \langle u^{\otimes 2}, w  u^{\otimes 2} \rangle \,d\nu(u) = -\log z_r.
\end{align}
\end{proof}

Since the interacting Gibbs measure $\mu$ is the unique minimizer for \eqref{eq:zr-rel}, we deduce from~\eqref{eq:rel-equal} 
$$\nu=\mu.$$
Moreover, we can remove the dependence of the subsequence $T_n$ in \eqref{eq:rel-equal} and \eqref{eq:weak_limit_DM_n} since the limiting objects are unique, and thus obtain the corresponding convergences for the whole family, namely 
\begin{align} \label{eq:rel-equal-whole-T}
\lim_{T \to \infty} \left( -\log \frac{Z_{\lambda,T}}{Z_{0,T}} \right) =  -\log z_r,
\end{align}
which is equivalent to \eqref{eq:rel-partition}, and 
\begin{equation}\label{eq:weak_limit_DM_whole_T}
\frac{k!}{T^k}\Gamma_{\lambda,T}^{(k)}\wto \int |u^{\otimes k}\rangle \langle u^{\otimes k}| \,d\mu(u)
\end{equation}
weakly  in $\gS ^2 (\gH _k)$ for every $k\ge 1$. To complete the proof of Theorem~\ref{thm:main} we now onmy need to upgrade the last convergence from weak to strong.

\subsection{Strong convergence of density matrices}

There remains to upgrade the weak convergence in~\eqref{eq:weak_limit_DM_whole_T} to the strong convergence. 

\medskip

\noindent{\bf Case $k=1$.} For the one-body density matrix, the strong convergence follows from the Dominated Convergence Theorem (for operators), the weak convergence in \eqref{eq:weak_limit_DM_n} and the following estimate in~\cite[Lemma 8.2]{LewNamRou-14d} (whose proof is based on a Feynman-Hellmann argument). 

\begin{lemma}[\textbf{Operator bound on the one-particle density matrix}]\label{lem:1PDM_bounded}\mbox{}\\
Let $h>0$ satisfy $h^{-1}\in\gS^p (\gH)$ for some $p\ge 1$ and let $w\geq0$ satisfy  
$$\tr_{\gH^2}\Big[w\,h^{-1}\otimes h^{-1}\Big]<\ii.$$ 
Then we have
\begin{equation}
0\leq \Gamma_{\lambda,T}^{(1)}\leq CTh^{-1}.
\label{eq:1PDM_bounded}
\end{equation}
\end{lemma}

\medskip

\noindent {\bf Case $k\ge 2$.} In this case an analogue of \eqref{eq:1PDM_bounded} is not available. Instead, we will use kernel estimates. Recall that from Proposition \ref{pro:new bounds} we know that 
\begin{equation}\label{eq:proof strong CV}
0\le \frac{\Gamma_{\lambda, T} ^{(k)} (X_k;Y_k)}{T^{k}(k!)} \leq C_k \frac{\Gamma_{0, T} ^{(k)} (X_k;Y_k)}{T^{k}(k!)}
\end{equation}
pointwise. Moreover, since $T^{-k}\Gamma_{0, T} ^{(k)}$ converges strongly to $\left( h ^{-1} \right) ^{\otimes k}$ in the Hilbert-Schmidt norm, its kernel converges strongly in $L^2$. It easily follows, using the Cauchy-Schwarz inequality, that 
\begin{equation}\label{eq:proof strong CV 2} 
\int_{\R^k \times \R ^k} \left| \frac{\Gamma_{0, T} ^{(k)} (X_k;Y_k) ^2}{T^{2k}(k!)^2} - \left( h ^{-1} \right) ^{\otimes k} (X_k;Y_k) ^2 \right| dX_k dY_k \underset{T\to \infty}{\longrightarrow} 0.
\end{equation}
The function $\left( h ^{-1} \right) ^{\otimes k} (X_k;Y_k) ^2$ is in $L^1 (\R^k \times \R ^k)$: it is positive and we easily check 
$$ \iint_{\R \times \R}  h ^{-1}   (x;y) ^2 dx dy = \tr\left[ h^{-2} \right] < \infty$$
by Proposition~\ref{lem:free}.
Therefore, if we can show that the kernel  $T^{-k} \Gamma_{\lambda, T} ^{(k)} (X_k;Y_k)$ converges pointwise, then it converges strongly in $L^2$ by  Lebesgue's Dominated Convergence Theorem (see the remark following~\cite[Theorem~1.8]{LieLos-01}). Then the operator $T^{-k}\Gamma_{\lambda, T} ^{(k)} $  will converge strongly in the Hilbert-Schmidt norm, as desired. 

To prove that the kernel $T^{-k} \Gamma_{\lambda, T} ^{(k)} (X_k;Y_k)$ converges pointwise, it suffices to show that the operator $T^{-k}\chi^{\otimes k}  \Gamma_{\lambda, T} ^{(k)} \chi^{\otimes k}$ converges strongly in the Hilbert-Schmidt norm when $\chi$ is a characteristic function of a ball. Indeed, we will prove a stronger statement

\begin{lemma}[\textbf{Local trace class convergence of density matrices}]\label{lem:strong-cv-chi-Gamma}\mbox{}\\
Let $\chi$ be the characteristic function of a ball. Then $T^{-k}\chi^{\otimes k}  \Gamma_{\lambda, T} ^{(k)} \chi^{\otimes k}$ converges strongly in the trace class for all $k\ge 1$.
\end{lemma}

\begin{proof} From the kernel estimate \eqref{eq:proof strong CV}, we have
$$
T^{-k}\Tr\Big[\chi^{\otimes k} \Gamma^{(k)}_{\lambda,T} \chi^{\otimes k} \Big] \le C  T^{-k} \Tr\Big[\chi^{\otimes k} \Gamma^{(k)}_{0,T} \chi^{\otimes k}\Big] \le C \left(\Tr\left[\chi h^{-1}\chi\right]\right)^k<\infty.
$$
Recall that we have shown during the proof of Lemma~\ref{lem:free supp} that $\Tr\left[\chi h^{-1}\chi\right]<\ii$ for $\chi$ a characteristic function. Thus $T^{-k}\chi^{\otimes k} \Gamma^{(k)}_{\lambda,T} \chi^{\otimes k}$ is bounded in trace class, and hence the weak convergence in \eqref{eq:weak_limit_DM_n} implies that
\begin{equation}
\frac{k!}{T^k}\chi^{\otimes k}\Gamma_{\lambda,T}^{(k)} \chi^{\otimes k} \wto \int_{\gH^{1-p}} |(\chi u)^{\otimes k}\rangle \langle (\chi u)^{\otimes k}| \,d\mu(u)
\label{eq:weak_limit_DM_n_abc}
\end{equation}
weakly-$\ast$ in trace-class norm\footnote{On the right side of \eqref{eq:weak_limit_DM_n_abc}, $\chi u \in L^2$ when $u\in \supp \, \mu \subset  \supp \, \mu_0  \subset L^4(\R)$.}. 

There remains to show that the convergence in \eqref{eq:weak_limit_DM_n_abc} is strong in the trace class. In the case $k=1$, the strong convergence again follows from the Dominated Convergence Theorem (for operators) and the operator bound from Lemma \ref{lem:1PDM_bounded}:
$$
0\le T^{-1}\chi  \Gamma_{\lambda,T}^{(1)} \chi \le C \chi h^{-1} \chi \in \gS^1 (\gH ).
$$
In the case $k\ge 2$, we use a general observation which has its own interest, Lemma~\ref{lem:cv-1-k} below. We postpone the proof of this result and finish that of Lemma \ref{lem:strong-cv-chi-Gamma}. Using the Fock space isomorphism 
$$
\cF(L^2(\R))=\cF\left(\chi L^2(\R) \oplus (1-\chi) L^2(\R) \right) \simeq \cF\left(\chi L^2(\R)\right)\otimes\cF\left((1-\chi) L^2(\R)\right)
$$
we can define the localized state $\widetilde \Gamma_{\lambda,T}$ on $\cF(\chi L^2(\R))$ by taking the partial trace of $\Gamma_{\lambda,T}$ over $\cF((1-\chi) L^2(\R))$. The density matrices of the localized state $\widetilde \Gamma_{\lambda,T}$ are given by
$$\left(\widetilde \Gamma_{\lambda,T}\right)^{(k)}=\chi^{\otimes k}\Gamma_{\lambda,T}^{(k)} \chi^{\otimes k},\qquad \forall k\geq1.$$
This localization procedure is well-known for many-particle quantum systems; see for instance~\cite[Appendix A]{HaiLewSol_2-09},~\cite{Lewin-11} or~\cite[Chapter~5]{Rougerie-LMU} for more detailed discussions. 

Applying Lemma \ref{lem:cv-1-k} with $(\eps_n,\Gamma_n)$ replaced by $(1/T,\widetilde \Gamma_{\lambda,T})$, we obtain the desired conclusion of Lemma \ref{lem:strong-cv-chi-Gamma}. 
\end{proof}

The general lemma we used above is as follows:

\begin{lemma}[\textbf{Strong convergence of higher density matrices}]\label{lem:cv-1-k}\mbox{}\\
Let $\gH$ be a separable Hilbert space and let $\{\Gamma_n\}$ be a sequence of states on the bosonic Fock space $\gF(\gH)$. Assume that there exists a sequence $0<\epsilon_n\to0$ and operators $\gamma^{(k)}$ such that 
\begin{equation}\label{eq:deF-trace}
(\eps_n)^k\,\Gamma_n^{(k)}\wto \gamma^{(k)}
\end{equation}
weakly-$\ast$ in trace class on $\bigotimes_{\rm sym}^k \gH$ for all $k\in \mathbb{N}$. If the convergence \eqref{eq:deF-trace} holds strongly in trace class for $k=1$, then it holds strongly in trace class for all $k\in \mathbb{N}$. 
\end{lemma}

The equivalent of this lemma for states with a fixed number of particles is a straightforward consequence of the weak quantum de Finetti theorem~\cite[Section~2]{LewNamRou-14}.

\begin{proof}
The strong convergence in \eqref{eq:deF-trace} follows from the fact that
\begin{equation} \label{eq:k-s-upper}
\limsup_{n\to \infty} (\eps_n)^k \Tr [\Gamma_n^{(k)}] \leq \Tr \gamma^{(k)}.
\end{equation}
We will show that if \eqref{eq:k-s-upper} holds for $k=1$, then it holds for all  $k \ge 2$. Let $0\le P \le 1$ be a finite rank projection on $\gH$ and let $Q=1-P$. We can decompose
\begin{align*}
\1^{\otimes k}&= Q\otimes \1^{\otimes k-1} + P\otimes \1^{\otimes k-1}    \\ 
&= Q\otimes \1^{\otimes k-1} + P\otimes Q \otimes \1^{\otimes k-2} + P^{\otimes 2} \otimes \1^{\otimes k-2} = \cdots \\  
&=Q \otimes \1^{\otimes k-1} + P\otimes Q \otimes \1^{\otimes k-2} + P^{\otimes 2} \otimes Q \otimes \1^{\otimes k-3} + \cdots + P^{\otimes k-1} \otimes Q + P^{\otimes k} \\
& \le Q \otimes \1^{\otimes k-1} + \1 \otimes Q \otimes \1^{\otimes k-2} + \1^{\otimes 2} \otimes Q \otimes \1^{\otimes k-3} + \cdots + \1^{\otimes k-1} \otimes Q + P^{\otimes k}.
\end{align*}
Therefore,
\begin{equation} \label{eq:Tr-Gk-decomp}
(\eps_n)^k \Tr \left[\Gamma_n^{(k)}\right]  \le  (\eps_n)^k \Tr \left[P^{\otimes k} \Gamma_n^{(k)}\right] + k (\eps_n)^k \Tr\left[ \left(Q \otimes \1^{\otimes k-1}\right)\Gamma_n^{(k)}\right].
\end{equation}
Now we estimate the right side of \eqref{eq:Tr-Gk-decomp}. The weak convergence in \eqref{eq:deF-trace} implies that 
\begin{equation} \label{eq:Tr-Gk-decomp-1}
\lim_{n\to \infty}  (\eps_n)^k \Tr \left[P^{\otimes k} \Gamma_n^{(k)}\right]  =  \Tr \left[P^{\otimes k} \gamma^{(k)}\right] \le \Tr \gamma^{(k)}.
\end{equation}
To estimate the second term on the right side of \eqref{eq:Tr-Gk-decomp}, we use the definition
$$
\Gamma_n^{(k)} = \sum_{m \ge k} {m \choose k} \Tr_{k+1\to m} G_n ^m
$$
with $G_{n} ^m$ the projection of $\Gamma_n$ onto $\bigotimes_{\rm sym}^m \gH$, namely 
$$
\Gamma_n = G_{n} ^0 \oplus G_{n} ^1 \oplus G_{n} ^2 \oplus \cdots,
$$
and $\Tr_{k+1\to m} G_{n} ^m$ is the partial trace of $G_{n} ^m$ with respect to $m-k$ variables\footnote{No matter which, by bosonic symmetry.}. In particular,
\begin{align*}
(\eps_n)^k \Tr\left[\left(Q \otimes \1^{\otimes k-1}\right)\Gamma_n^{(k)}\right] &=  (\eps_n)^k \sum_{m\ge k} {m \choose k} \Tr\left[ \left(Q\otimes \1\right) \Tr_{2\to m} G_{n} ^m \right] \\
&\le  \sum_{m\ge 2} (\eps_n m)^k  \Tr\left[ \left(Q\otimes \1\right) \Tr_{2\to m} G_{n} ^m \right]. 
\end{align*}
Let $M\ge 1$ and divide the sum into two parts: $\eps_n m \le M$ and $\eps_n m>M$. Then, using
$$
 (\eps_n m)^k \le \left\{
 \begin{aligned}
 M^{k-1} (\eps_n m )& \quad \text{ if \quad $\eps_n m \le M$,} \\
 M^{-1} (\eps_n m)^{k+1} \quad & \quad\text{ if \quad $\eps_n m>M$,}
 \end{aligned}
 \right.
$$
we can estimate
\begin{align*}
&(\eps_n)^k \Tr\left[(Q \otimes \1^{\otimes k-1})\Gamma_n^{(k)}\right] \\
& \le M^{k-1}  \sum_{m \ge 2} (\eps_n m)  \Tr\left[ (Q\otimes \1) \Tr_{2\to m} G_{n} ^m \right]  +  M^{-1} \sum_{m \ge 2} (\eps_n m)^{k+1}  \Tr\left[ G_{n} ^m \right] \\
&\le  M^{k-1} \eps_n \Tr \left[Q \Gamma_{n}^{(1)}\right] +  M^{-1} \Tr \left[ \left(\eps_n \cN\right)^{k+1} \Gamma_n\right].
\end{align*}
Here $\cN$ is the usual number operator on the Fock space $\gH$. Since $\eps_n \Gamma_n^{(1)}$ converges strongly in trace class, we get
$$
\lim_{n\to \infty} \eps_n \Tr \left[Q \Gamma_{n}^{(1)}\right] =  \Tr \left[ Q \gamma^{(1)} \right].
$$
On the other hand, since $(\eps_n)^\ell \Gamma_n^{(\ell)}$ converges weakly-$\ast$ in trace class, its trace is bounded uniformly in $n$. Combining with the identity
$$
\Tr  \Gamma_n^{(\ell)}= \Tr_{\cF(\gH)} \left[ {\cN \choose \ell} \Gamma_n \right] , \quad \forall \ell\ge 1
$$
we find that 
$$
\limsup_{n\to \infty} \Tr \left[ (\eps_n \cN)^{\ell} \Gamma_n\right] \le C_\ell, \quad \forall \ell\ge 1
$$
for a constant $C_\ell$ independent of $n$. Thus we have shown that
\begin{equation} \label{eq:Tr-Gk-decomp-2}
\limsup_{n\to \infty} (\eps_n)^k \Tr\left[(Q \otimes \1^{\otimes k-1})\Gamma_n^{(k)}\right] \le M^{k-1}  \Tr \left[Q \gamma^{(1)} \right] + \frac{C_k}{M}.
\end{equation}
In summary, inserting \eqref{eq:Tr-Gk-decomp-1} and \eqref{eq:Tr-Gk-decomp-2} into \eqref{eq:Tr-Gk-decomp} we obtain
\begin{align*}
\limsup_{n\to \infty} (\eps_n)^k \Tr \left[\Gamma_n^{(k)}\right]  \le  \Tr \gamma^{(k)} + k M^{k-1} \Tr \left[ Q \gamma^{(1)}\right]+ \frac{kC_k}{M}
\end{align*}
for all projections $Q$, all $M\ge 1$ and all $k\ge 2$. It remains to take $P \to \1$, then $M\to \infty$, to conclude that   
\begin{align*}
\limsup_{n\to \infty} (\eps_n)^k \Tr [\Gamma_n^{(k)}]  \le \Tr \gamma^{(k)} .
\end{align*}
The proof is complete.  
\end{proof}

By the same proof, we can show that if \eqref{eq:deF-trace} holds weakly-$\ast$ in trace class for all $1\le k \le \kappa$  and strongly in trace class for $k=1$, then it holds strongly in trace class for all $1\le k \le \kappa-1$.

%

\begin{thebibliography}{10}

\bibitem{AmmNie-08}
{\sc Z.~Ammari and F.~Nier}, {\em Mean field limit for bosons and infinite
  dimensional phase-space analysis}, Ann. Henri Poincar\'e, 9 (2008),
  pp.~1503--1574.

\bibitem{Berezin-72}
{\sc F.~A. Berezin}, {\em Convex functions of operators}, Mat. Sb. (N.S.),
  88(130) (1972), pp.~268--276.

\bibitem{Bhatia}
{\sc R.~Bhatia}, {\em Matrix analysis}, vol.~169, Springer, 1997.

\bibitem{Bourgain-94}
{\sc J.~Bourgain}, {\em Periodic nonlinear {S}chr\"odinger equation and
  invariant measures}, Comm. Math. Phys., 166 (1994), pp.~1--26.

\bibitem{Bourgain-96}
\leavevmode\vrule height 2pt depth -1.6pt width 23pt, {\em Invariant measures
  for the {2D}-defocusing nonlinear {S}chr\"odinger equation}, Comm. Math.
  Phys., 176 (1996), pp.~421--445.

\bibitem{Bourgain-97}
\leavevmode\vrule height 2pt depth -1.6pt width 23pt, {\em {Invariant measures
  for the Gross-Pitaevskii equation}}, Journal de Math\'ematiques Pures et
  Appliqu\'ees, 76 (1997), pp.~649--02.

\bibitem{BouBul-14a}
{\sc J.~Bourgain and A.~Bulut}, {\em {Almost sure global well posedness for the
  radial nonlinear Schr\"odinger equation on the unit ball I: the 2D case}},
  Annales I. H. Poincare (C), 31 (2014), pp.~1267--1288.

\bibitem{BouBul-14b}
\leavevmode\vrule height 2pt depth -1.6pt width 23pt, {\em {Almost sure global
  well posedness for the radial nonlinear Schr\"odinger equation on the unit
  ball II: the 3D case}}, Journal of the European Mathematical Society, 16
  (2014), pp.~1289--1325.

\bibitem{BraRob2}
{\sc O.~Bratelli and D.~W. Robinson}, {\em Operator Algebras and Quantum
  Statistical Mechanics 2: Equilibrium States. Models in Quantum Statistical
  Mechanics}, Texts and Monographs in Physics, Springer, 2nd~ed., 2002.

\bibitem{BurThoTzv-09}
{\sc N.~{Burq}, L.~{Thomann}, and N.~{Tzvetkov}}, {\em {Gibbs measures for the
  non linear harmonic oscillator}}, in Journ\'ees EDP ï¿½vian 2009., 2009.

\bibitem{BurThoTzv-10}
\leavevmode\vrule height 2pt depth -1.6pt width 23pt, {\em {Long time dynamics
  for the one dimensional non linear Schr\"odinger equation}}, Ann. Inst.
  Fourier., 63 (2013), pp.~2137--2198.

\bibitem{BurTzv-08}
{\sc N.~Burq and N.~Tzvetkov}, {\em Random data {C}auchy theory for
  supercritical wave equations. {I}. {L}ocal theory}, Invent. Math., 173
  (2008), pp.~449--475.

\bibitem{CacSuz-14}
{\sc F.~Cacciafesta and A.-S. {de Suzzoni}}, {\em Invariant measure for the
  {S}chr\"odinger equation on the real line}, J. Func Anal, 269 (2015),
  pp.~271--324.

\bibitem{PraDeb-03}
{\sc G.~da~Prato and A.~Debbussche}, {\em Strong solutions to the stochastic
  quantization equations}, Ann. Probab., 32 (2003), pp.~1900--1916.

\bibitem{BouDebFuk-17}
{\sc A.~de~Bouard, A.~Debussche, and R.~Fukuizumi}, {\em Long time behavior of
  {G}ross-{P}itaevskii equation at positive temperature}.
\newblock arXiv:1708.01961, 2017.

\bibitem{DerGer-13}
{\sc J.~Derezi{\'n}ski and C.~G{\'e}rard}, {\em {Mathematics of Quantization
  and Quantum Fields}}, Cambridge University Press, Cambridge, 2013.

\bibitem{DolFelLosPat-06}
{\sc J.~Dolbeault, P.~Felmer, M.~Loss, and E.~Paturel}, {\em Lieb-{T}hirring
  type inequalities and {G}agliardo-{N}irenberg inequalities for systems}, J.
  Funct. Anal., 238 (2006), pp.~193--220.

\bibitem{Fernique-70}
{\sc X.~Fernique}, {\em Int\'egrabilit\'e des vecteurs gaussiens}, C. R. Acad.
  Sci. Paris S\'er. A-B, 270 (1970), pp.~A1698--A1699.

\bibitem{FroKnoSchSoh-16}
{\sc J.~Fr\"ohlich, A.~Knowles, B.~Schlein, and V.~Sohinger}, {\em Gibbs
  measures of nonlinear {S}chr\"odinger equations as limits of quantum
  many-body states in dimensions $d\leq 3$}.
\newblock arXiv:1605.07095, 2016.

\bibitem{FroKnoSchSoh-17}
\leavevmode\vrule height 2pt depth -1.6pt width 23pt, {\em A microscopic
  derivation of time-dependent correlation functions of the {1D} cubic
  nonlinear {Schr\"odinger} equation}.
\newblock arXiv:1703.04465, 2017.

\bibitem{FroPar-78}
{\sc J.~Fr\"ohlich and Y.~M. Park}, {\em {Correlation Inequalities and the
  Thermodynamic Limit for Classical and Quantum Continuous Systems}}, Commun.
  Math. Phys., 59 (1978), pp.~235--266.

\bibitem{FroPar-80}
\leavevmode\vrule height 2pt depth -1.6pt width 23pt, {\em {Correlation
  Inequalities and the Thermodynamic Limit for Classical and Quantum Continuous
  Systems II. Bose-Einstein and Fermi-Dirac Statistics}}, J. Stat. Phys., 23
  (1980), pp.~701--753.

\bibitem{Ginibre-65b}
{\sc J.~Ginibre}, {\em {Reduced density matrices for quantum gases I. Limit of
  infinite volume}}, J. Math. Phys., 6 (1965), pp.~238--251.

\bibitem{Ginibre-65c}
\leavevmode\vrule height 2pt depth -1.6pt width 23pt, {\em {Reduced density
  matrices for quantum gases II. Cluster Property}}, J. Math. Phys., 6 (1965),
  pp.~252--262.

\bibitem{Ginibre-65d}
\leavevmode\vrule height 2pt depth -1.6pt width 23pt, {\em {Reduced density
  matrices for quantum gases II. Hard-core potentials}}, J. Math. Phys., 6
  (1965), pp.~1432--1446.

\bibitem{GliJaf-87}
{\sc J.~Glimm and A.~Jaffe}, {\em Quantum Physics: A Functional Integral Point
  of View}, Springer-Verlag, 1987.

\bibitem{HaiLewSol_2-09}
{\sc C.~Hainzl, M.~Lewin, and J.~P. Solovej}, {\em The thermodynamic limit of
  quantum {C}oulomb systems. {P}art {II}. {A}pplications}, Advances in Math.,
  221 (2009), pp.~488--546.

\bibitem{KocTat-05}
{\sc H.~Koch and D.~Tataru}, {\em {$L^p$ eigenfunction bounds for the Hermite
  operator}}, Duke Math. J., 128 (2005), pp.~369--392.

\bibitem{KocTatZwo-07}
{\sc H.~Koch, D.~Tataru, and M.~Zworski}, {\em {Semiclassical $L^p$
  estimates}}, Ann. Henri Poincar\'e, 8 (2007), pp.~885--916.

\bibitem{LebRosSpe-88}
{\sc J.~L. Lebowitz, H.~A. Rose, and E.~R. Speer}, {\em Statistical mechanics
  of the nonlinear {S}chr\"odinger equation}, J. Statist. Phys., 50 (1988),
  pp.~657--687.

\bibitem{Lewin-11}
{\sc M.~Lewin}, {\em Geometric methods for nonlinear many-body quantum
  systems}, J. Funct. Anal., 260 (2011), pp.~3535--3595.

\bibitem{LewNamRou-ICMP}
{\sc M.~Lewin, P.~Nam, and N.~Rougerie}, {\em {Bose gases at positive
  temperature and non-linar Gibbs measures}}, Preprint (2016) arXiv:1602.05166.

\bibitem{LewNamRou-14}
\leavevmode\vrule height 2pt depth -1.6pt width 23pt, {\em Derivation of
  {H}artree's theory for generic mean-field {B}ose systems}, Adv. Math., 254
  (2014), pp.~570--621.

\bibitem{LewNamRou-14d}
\leavevmode\vrule height 2pt depth -1.6pt width 23pt, {\em Derivation of
  nonlinear {G}ibbs measures from many-body quantum mechanics}, Journal de
  l'Ecole Polytechnique, 2 (2016), pp.~553--606.

\bibitem{Lieb-73b}
{\sc E.~H. Lieb}, {\em The classical limit of quantum spin systems}, Comm.
  Math. Phys., 31 (1973), pp.~327--340.

\bibitem{LieLos-01}
{\sc E.~H. Lieb and M.~Loss}, {\em Analysis}, vol.~14 of Graduate Studies in
  Mathematics, American Mathematical Society, Providence, RI, 2nd~ed., 2001.

\bibitem{OhyPet-93}
{\sc M.~Ohya and D.~Petz}, {\em Quantum entropy and its use}, Texts and
  Monographs in Physics, Springer-Verlag, Berlin, 1993.

\bibitem{ReeSim1}
{\sc M.~Reed and B.~Simon}, {\em Methods of {M}odern {M}athematical {P}hysics.
  {I}. Functional analysis}, Academic Press, 1972.

\bibitem{ReeSim2}
\leavevmode\vrule height 2pt depth -1.6pt width 23pt, {\em Methods of {M}odern
  {M}athematical {P}hysics. {II}. {F}ourier analysis, self-adjointness},
  Academic Press, New York, 1975.

\bibitem{RocZhuZhu-16}
{\sc M.~Rockner, R.~Zhu, and X.~Zhu}, {\em Ergodicity for the stochastic
  quantization problems on the {2D}-torus}.
\newblock arXiv:1606.02102, 2016.

\bibitem{Rougerie-LMU}
{\sc N.~Rougerie}, {\em {De Finetti theorems, mean-field limits and
  Bose-Einstein condensation}}.
\newblock arXiv:1506.05263, 2014.
\newblock LMU lecture notes.

\bibitem{Rougerie-cdf}
\leavevmode\vrule height 2pt depth -1.6pt width 23pt, {\em Th\'eor\`emes de de
  {F}inetti, limites de champ moyen et condensation de {B}ose-{E}instein}.
\newblock arXiv:1409.1182, 2014.
\newblock Lecture notes for a cours Peccot.

\bibitem{Rougerie-xedp15}
\leavevmode\vrule height 2pt depth -1.6pt width 23pt, {\em {From bosonic
  grand-canonical ensembles to nonlinear {G}ibbs measures}}, 2014-2015.
\newblock S{\'e}minaire Laurent Schwartz.

\bibitem{Simon-74}
{\sc B.~Simon}, {\em The {$P(\Phi )_{2}$} {E}uclidean (quantum) field theory},
  Princeton University Press, Princeton, N.J., 1974.
\newblock Princeton Series in Physics.

\bibitem{Simon-79}
\leavevmode\vrule height 2pt depth -1.6pt width 23pt, {\em Trace ideals and
  their applications}, vol.~35 of London Mathematical Society Lecture Note
  Series, Cambridge University Press, Cambridge, 1979.

\bibitem{Simon-80}
\leavevmode\vrule height 2pt depth -1.6pt width 23pt, {\em The classical limit
  of quantum partition functions}, Comm. Math. Phys., 71 (1980), pp.~247--276.

\bibitem{Simon-05}
\leavevmode\vrule height 2pt depth -1.6pt width 23pt, {\em Functional
  integration and quantum physics}, AMS Chelsea Publishing, Providence, RI,
  second~ed., 2005.

\bibitem{Skorokhod-74}
{\sc A.~Skorokhod}, {\em Integration in {H}ilbert space}, Ergebnisse der
  Mathematik und ihrer Grenzgebiete, Springer-Verlag, 1974.

\bibitem{ThoOh-15}
{\sc L.~Thomann and T.~Oh}, {\em {Invariant Gibbs measures for the 2D
  defocusing nonlinear Schr\"odinger equations}}.
\newblock arXiv:1509.02093, 2015.

\bibitem{ThoTzv-10}
{\sc L.~Thomann and N.~Tzvetkov}, {\em Gibbs measure for the periodic
  derivative nonlinear {S}chr\"odinger equation}, Nonlinearity, 23 (2010),
  p.~2771.

\bibitem{TsaWeb-16}
{\sc P.~Tsatsoulis and H.~Weber}, {\em Spectral gap for the stochastic
  quantization equation on the 2-dimensional torus}.
\newblock arXiv:1609.08447, 2016.

\bibitem{Tzvetkov-08}
{\sc N.~Tzvetkov}, {\em Invariant measures for the defocusing nonlinear
  {S}chr\"odinger equation}, Ann. Inst. Fourier (Grenoble), 58 (2008),
  pp.~2543--2604.

\bibitem{VelWig-73}
{\sc G.~Velo and A.~Wightman}, eds., {\em {Constructive quantum field theory:
  The 1973 Ettore Majorana international school of mathematical physics}},
  Lecture notes in physics, Springer-Verlag, 1973.

\bibitem{Wehrl-78}
{\sc A.~Wehrl}, {\em General properties of entropy}, Rev. Modern Phys., 50
  (1978), pp.~221--260.

\bibitem{YajZha-01}
{\sc K.~Yajima and G.~Zhang}, {\em {Smoothing property for Schrï¿½dinger
  equations with potential superquadratic at infinity}}, Comm. Math. Phys., 221
  (2001), pp.~537--590.

\end{thebibliography}
%

\end{document}